\documentclass[11pt]{article}
\usepackage{amsmath, amsthm, amssymb, bbm, systeme}
\usepackage[margin=1in]{geometry}
\usepackage[round]{natbib}

\theoremstyle{plain}
\newtheorem{theorem}{Theorem}[section]
\newtheorem{lemma}[theorem]{Lemma}

\theoremstyle{definition}

\theoremstyle{remark}
\newtheorem{example}[theorem]{Example}

\newcommand{\eps}{\varepsilon}

\newcommand{\cD}{\mathcal{D}}
\newcommand{\ESPREV}{\textsc{EspRev}}
\newcommand{\OPT}{\textsc{MyeRev}}
\newcommand{\asp}{\textsc{ASP}}
\newcommand{\esp}{\textsc{ESP}}
\newcommand{\lsp}{\textsc{LSP}}
\newcommand{\spm}{\textsc{SPM}}

\begin{document}
\title{Separation between Second Price Auctions with Personalized Reserves and the Revenue Optimal Auction
}
\author{Will Ma\thanks{Google Research. \tt{willma@google.com}.}
 \and Balasubramanian Sivan\thanks{Google Research. \tt{balusivan@google.com}.}
}
\date{}
\maketitle{}

\begin{abstract}
What fraction of the single item $n$ buyers setting's expected optimal revenue (\OPT) can the second price auction with reserves achieve? In the special case where the buyers' valuation distributions are all drawn i.i.d. and the distributions satisfy the regularity condition, the second price auction with an anonymous reserve (\asp) is the optimal auction itself.  As the setting gets more complex, there are established upper bounds on the fraction of \OPT\ that \asp\ can achieve. On the contrary, no such upper bounds are known for the fraction of \OPT\ achievable by the second price auction with \emph{eager personalized reserves} (\esp). In particular, no separation was earlier known between \esp's revenue and \OPT\ even in the most general setting of non-identical product distributions that don't satisfy the regularity condition. In this paper we establish the first separation results for \esp: we show that even in the case of distributions drawn i.i.d., but not necessarily satisfying the regularity condition, the \esp\ cannot achieve more than a $0.778$ fraction of \OPT\ in general. Combined with the result from~\cite{correa2017posted} that \esp\ can achieve at least a $0.745$ fraction of \OPT, this nearly bridges the gap between upper and lower bounds on \esp's approximation factor.
\end{abstract}

\section{Introduction}
Approximating the optimal auction's expected revenue \OPT\ by simple mechanisms has received a lot of attention in the algorithmic game theory literature, especially in the last decade. The seminal result of~\cite{M81} describes the optimal auction in a single agent $n$-items setting when buyer valuations are drawn from a product distribution. When the buyer valuation distributions are independent and identical (i.i.d.) and satisfy a technical regularity condition, Myerson's optimal auction coincides with the second price (Vickrey's) auction with an anonymous reserve price (\asp). When the setting gets more complex, either because the distributions are not identical or because the distributions don't satisfy the regularity condition, or both, \asp\ ceases to be the optimal auction.  Bounds on the suboptimality are known in many cases:
\begin{itemize}
    \item when the distributions are non-identical (independent) and irregular,~\cite{AHNPY15} show that \asp\ cannot get more than a $\frac{1}{n}$ fraction of \OPT\ in general, and it is straightforward to show that this is tight;
    \item when the distributions are non-identical (independent) and regular,~\cite{HR09} show that \asp\ cannot get more than a $\frac{1}{2}$ fraction of \OPT\ in general,~\cite{AHNPY15} show that \asp\ can achieve a $\frac{1}{e}$ fraction of \OPT\ in general, and recently~\cite{jin2019tight} show that \asp\ cannot get more than a $\frac{1}{2.15}$ fraction of \OPT\ in general;
    \item when the distributions are identical (independent) and irregular, the pricing problem can be reduced to the prophet inequalities problem on ironed virtual values, for which a $\frac{1}{2}$ approximation via a uniform threshold is known (\cite{samuel-cahn1984}). Since the distributions are identical, the uniform threshold in the ironed virtual value space translates into a uniform threshold in the value space, thus showing that ASP can achieve a $\frac{1}{2}$ approximation as well.  This approximation factor of \asp\ is tight, as shown in \citet{hartline2013mechanism}, although we will also provide a self-contained proof here as a corollary.
\end{itemize}

On the contrary, \emph{no upper bounds are known} on the fraction of \OPT\ obtainable by a second price auction with personalized reserve prices.
When the reserve prices are personalized, it matters whether the bidders are ranked after eliminating those below the reserve (eager, ``\esp") or before eliminating those below the reserve (lazy, ``\lsp").
In the eager version, we first eliminate bidders below their respective reserves, and then charge the highest surviving bidder (if any) the maximum of his own reserve price and the second highest surviving bid.
In the lazy version, we first decide the potential winner to be the highest bidder, eliminate him if he doesn't survive his reserve price (in which case we allocate the good to nobody), and then if he survives allocate the good to him at a price of maximum of his reserve price and the second highest bid.~\cite{PPV16} show that for general correlated distributions, the optimal revenue obtainable from \esp\ and \lsp\ are incomparable (i.e.\ either could be higher on a specific instance), but are within a factor of two. When the distributions are independent~\citeauthor{PPV16} show that the optimal \esp's revenue dominates the optimal \lsp's revenue. Therefore from the perspective of upper bounding the fraction of \OPT\ that is obtainable, 
any result for \esp\ implies a result for \lsp.
ESP's are also more widely used, and are a more natural auction to run\footnote{For instance, when there is at least one buyer that survives their reserve price, \esp\ always allocates the item, whereas \lsp\ could result in no allocation even then.}).

\paragraph{Our result.} Our main result (Theorem~\ref{thm:mainthm}) is an upper bound on the maximum fraction of \OPT\ that \esp\ can obtain. We show that in the $n$ buyers single item setting, even when the buyer values are drawn i.i.d. (but not necessarily from a regular distribution), the fraction of \OPT\ that \esp\ can extract in general is at most $0.778$. The i.i.d. irregular setting is the simplest non-trivial setting to consider since, as discussed earlier, the i.i.d. regular setting results in a $1$ approximation from even a second price auction with anonymous reserve price. 

\paragraph{Corollary.} Our construction can also be modified to show that \asp\ cannot get more than $\frac{1}{2}$ of \OPT\ for i.i.d. irregular valuations, establishing that the $\frac{1}{2}$ approximation factor for \asp\ is tight.

\paragraph{Near tightness of ESP result.} A mechanism very related to \esp\ is the sequential posted price mechanism (\spm) where the seller computes one price per buyer, and then approaches the buyers in the descending order of posted prices until the item is sold. In general single-parameter environments with matroid feasibility constraints (this includes the single item setting we study as a special case),~\cite{CHMS10} showed the the revenue of an \esp\, that uses the posted prices of an \spm\ as its personalized reserve prices is at least as high as the \spm's revenue. This, combined with the recent result of~\cite{correa2017posted} showing that \spm\ can achieve a $0.745$ fraction of \OPT\ in the i.i.d. setting with possibly irregular distributions, shows that \esp\ can also achieve at least a $0.745$ approximation factor in the i.i.d. setting. Thus, our upper bound of $0.778$ on \esp's approximation factor in the i.i.d. setting is quite close to this lower bound of $0.745$. In the case of \spm\, both the upper and lower bound are known to be $0.745$: the upper bound of $0.745$ for \spm s was established by~\cite{HK82} about $3.5$ decades prior to the matching lower bound established by~\citeauthor{correa2017posted}

\paragraph{\esp\ vs \spm: the benefit of simultaneity.} The primary advantage of \esp\ over \spm\ is that \esp\ considers all buyers simultaneously while \spm\ considers buyers sequentially. Thus, as~\cite{CHMS10} show, \esp\ should only earn better revenue by benefiting from this simultaneity. Our upper bound of $0.778$ on the approximation factor of \esp\ shows that this benefit of simultaneity is quite small in general, given that \spm s themselves can achieve a $0.745$ approximation. In a sense, our result shows that \emph{discrimination} is more crucial for revenue maximization than \emph{simultaneity}. Closing this gap further, or establishing a separation between \spm\ and \esp\ remains an interesting open problem. 

\section{Upper bound on ESP approximation factor}
We provide an example that establishes the upper bound on \esp's approximation factor.
\begin{example} \label{eg::twoPoint}
There are $n$ buyers, with $n\to\infty$.  Each buyer's valuation is drawn i.i.d. from a discrete distribution $F$ where
\begin{align*}
v\sim F = \begin{cases}
n, &\text{w.p. $1/n^2$} \\
\alpha/\beta, &\text{w.p. $\beta/n - 1/n^2$} \\
0, &\text{w.p. $1 - \beta/n$}
\end{cases}
\end{align*}
Here $\alpha > 1$ and $\beta>0$ are constants to be optimized later to minimize the ratio of \OPT\ earned by \esp\ (note that since they are constants they don't scale with $n$, and therefore $\beta >> \frac{1}{n}$ and $\alpha < n\beta$; these are things we use in Lemma~\ref{prop::mye} below).
\end{example}

\begin{lemma} \label{prop::mye}
In a $n$ buyers setting, when the buyer valuations are drawn i.i.d. from the distribution $F$ described in Example~\ref{eg::twoPoint}, the expected revenue \OPT\ of Myerson's optimal auction approaches
\begin{align*}
1+\frac{\alpha-1}{\beta}(1-e^{-\beta}).
\end{align*}
\end{lemma}

\begin{proof}
We compute the revenue as the expected ironed virtual surplus. To compute the latter, consider the revenue curve $R(q)=qF^{-1}(1-q)$, where in this case $F^{-1}(1-q)$ denotes the
maximum
price at which the probability of getting a sale (from a single buyer with valuation drawn from $F$) is at least $q$.  It can be computed that
\begin{align*}
R(q)=qF^{-1}(1-q)=\begin{cases}
qn, &0<q\le 1/n^2; \\
q\alpha/\beta, &1/n^2<q\le\beta/n; \\
0, &\beta/n<q\le1.
\end{cases}
\end{align*}
Note that the revenue curve is not concave, because $R(1/n^2)=1/n$ and $R(\beta/n)=\alpha/n>1/n$, while $R(1/n^2+\eps) = (\frac{1}{n^2}+\epsilon)\frac{\alpha}{\beta}<1/n$ for sufficiently small $\eps>0$.
The ironed revenue curve over $q\in[1/n^2,\beta/n]$ is formed by joining the points $(1/n^2,1/n)$ and $(\beta/n,\alpha/n)$. The slope of three line segments in the ironed revenue curve are the ironed virtual values corresponding to the three value realizations:
\begin{itemize}
    \item the slope of the line segment between $(0,0)$ and $(1/n^2,1/n)$ is $n$ and that is the virtual value (which is equal to the ironed virtual value for this segment) of buyers with value $n$;
    \item the slope of the line segment between $(1/n^2,1/n)$ and $(\beta/n,\alpha/n)$ is $\frac{\alpha-1}{\beta-1/n}$
and that is the ironed virtual valuation of buyers with value $\alpha/\beta$.
    \item buyers with value $0$ have negative ironed virtual value and are irrelevant for revenue purposes.
\end{itemize}
 
The complete description of Myerson's optimal mechanism here is as follows:
\begin{itemize}
    \item any buyer who bids below $\frac{\alpha}{\beta}$ is never allocated, and always pays 0;
    \item if there are no buyers with bid $n$ or larger: allocate the item to a uniformly random buyer among all buyers with bids in $[\frac{\alpha}{\beta},n)$ and charge him $\frac{\alpha}{\beta}$;
    \item if there is exactly one buyer with bid of $n$ or larger, that buyer gets the item with probability $1$, and he pays $n - (n-\frac{\alpha}{\beta})\cdot 1/(1 + \text{num of buyers with bids in } [\frac{\alpha}{\beta},n))$ which is basically the area to the left of the allocation curve of that buyer;
    \item if there is more than one buyer with bid of n or larger, allocate the item to a uniformly random buyer among all buyers with bids in $[n, \infty)$ and charge him $n$. 
\end{itemize}
The expected revenue of Myerson's optimal mechanism \OPT\ equals the expected ironed virtual surplus, which can be computed as follows.
With probability $1-(1-1/n^2)^n$, at least one buyer has virtual valuation $n$.
With probability $(1-1/n^2)^n-(1-\beta/n)^n$, no buyer has virtual valuation $n$ but at least one buyer has ironed virtual valuation $\frac{\alpha-1}{\beta - 1/n}$.
Therefore, the expected ironed virtual surplus, as the number of buyers approaches $\infty$, is
\begin{align*}
\lim_{n\to\infty}\left(n\Big(1-(1-\frac{1}{n^2})^n\Big)+\frac{\alpha-1}{\beta - 1/n}\Big((1-\frac{1}{n^2})^n-(1-\frac{\beta}{n})^n\Big)\right)=1+\frac{\alpha-1}{\beta}(1-e^{-\beta}),
\end{align*}
completing the proof.
\end{proof}

\begin{lemma} \label{prop::esp}
In a $n$ buyers setting, when the buyer valuations are drawn i.i.d. from the distribution $F$ described in Example~\ref{eg::twoPoint}, the expected revenue \ESPREV\ of \esp\ as $n\to\infty$ is at most 
\begin{align*}
1+\frac{\alpha-\ln\alpha-1}{\beta}.
\end{align*}
\end{lemma}

\begin{proof}
Consider any revenue-maximizing ESP auction.  Since the only possible non-zero valuations are $n$ and $\alpha/\beta$, it does not lose generality to assume that all reserve prices are set to either $n$ or $\alpha/\beta$.
Let $z\in\{0,\frac{1}{n},\frac{2}{n},\ldots,1\}$ denote the fraction of buyers whose reserve is set to $n$, so that the number of buyers with reserves $n$ and $\alpha/\beta$ are $zn$ and $(1-z)n$, respectively.

With probability $1-(1-\frac{1}{n^2})^{zn}$, at least one of the buyers with reserve $n$ clears their reserve price, and the ESP auction earns the maximum possible revenue of $n$.
Otherwise, all $zn$ buyers with reserve $n$ are eliminated from the auction and we consider the other buyers.

For the other $(1-z)n$ buyers with reserve $\alpha/\beta$,
at least one of them will clear their reserve with probability $1-(1-\frac{\beta}{n})^{(1-z)n}$,
in which case the ESP auction is guaranteed to earn at least $\alpha/\beta$.
In this case, there is also the possibility of earning the higher revenue of $n$, if at least two buyers report a valuation of $n$.
However, the probability of this event is at most $\binom{n}{2}$ times the probability that a \textit{given} pair of buyers both report $n$, by the union bound.
Therefore, the expected revenue collected from this event is at most $n\binom{n}{2}(\frac{1}{n^2})^2=O(\frac{1}{n})$, and can be ignored as $n\to\infty$.

Collecting the above terms, for any $n$ and $z\in\{0,\frac{1}{n},\frac{2}{n},\ldots,1\}$, the expected revenue of the corresponding ESP auction is
\begin{align}
&n\Big(1-(1-\frac{1}{n^2})^{zn}\Big)+\frac{\alpha}{\beta}(1-\frac{1}{n^2})^{zn}\Big(1-(1-\frac{\beta}{n})^{(1-z)n}\Big) \nonumber \\
\le\ &n\Big(1-(1-\binom{zn}{1}\frac{1}{n^2})\Big)+\frac{\alpha}{\beta}(1)\Big(1-e^{-\beta(1-z)}\Big) \nonumber \\
= &z+\frac{\alpha}{\beta}\Big(1-e^{-\beta(1-z)}\Big)+O(\frac{1}{n}) \label{eqn:espUB}
\end{align}
where we used the binomial theorem in the first inequality. The maximum of the concave function $z+\frac{\alpha}{\beta}(1-e^{-\beta(1-z)})$ over $z\in[0,1]$ is attained at $z=1-\frac{\ln\alpha}{\beta}$, and equals $1-\frac{\ln\alpha}{\beta}+\frac{\alpha}{\beta}(1-\frac{1}{\alpha})$, completing the proof.
\end{proof}

\begin{theorem}
\label{thm:mainthm}
Let $\OPT(n,\cD)$ and $\ESPREV(n,\cD)$ denote respectively the revenue of the optimal auction (Myerson's auction) and the second price auction with personalized reserves in a single-item setting with $n$ buyers whose values are drawn i.i.d from $\cD$. Then:
\begin{align*}
\inf_{n,\cD}\frac{\ESPREV(n,\cD)}{\OPT(n,\cD)}<0.778.
\end{align*}
\end{theorem}

\begin{proof}
Consider Example~\ref{eg::twoPoint}.  By Lemmas~\ref{prop::mye} and~\ref{prop::esp}, the fraction of Myerson's optimal revenue earned by an ESP auction in this example is at most
\begin{align} \label{eqn::ratio}
\frac{\beta+\alpha-1-\ln\alpha}{\beta+(\alpha-1)(1-e^{-\beta})},
\end{align}
for any values of $\alpha>1$ and $\beta>0$.  By numerically minimizing over such values of $\alpha$ and $\beta$, it can be checked that when $\alpha=2.91$ and $\beta=1.89$, the value of (\ref{eqn::ratio}) is less than 0.778, completing the proof.
\end{proof}

\subsection{Upper bound on ASP approximation factor}

Consider the distribution from Example~\ref{eg::twoPoint} with $\alpha=\beta=n$ and let $n\to\infty$.
By Lemma~\ref{prop::mye}, the optimal revenue of Myerson is $1+\frac{n-1}{n}(1-e^{-n})$ which approaches 2 as $n\to\infty$.
On the other hand, \esp\ must set all of the reserves to either $n$ or 1.
By the analysis in Lemma~\ref{prop::esp}, this is upper-bounded by expression~\eqref{eqn:espUB} with $z=0$ or $z=1$, i.e. the revenue of \asp\ is at most
\begin{align*}
\max\{1-e^{-n}+O(\frac{1}{n}),1+O(\frac{1}{n})\}
\end{align*}
which approaches 1 as $n\to\infty$.  This completes the proof.

\bibliographystyle{abbrvnat}
\bibliography{main.bib}
\end{document}